\documentclass[letterpaper,11pt]{article}
\usepackage{amsmath, amssymb}
\usepackage{color}
\usepackage{fullpage}

\usepackage[noend]{algorithmic}
\usepackage{algorithm}

\usepackage{enumerate}

\usepackage{hyperref}

\usepackage{graphicx}
\usepackage{caption}
\usepackage{subcaption}
\usepackage{zerosum,csquotes}
\usepackage{enumitem}

 \newtheorem{theorem}{Theorem}[section]
 \newtheorem{lemma}[theorem]{Lemma}

 \newtheorem{claim}[theorem]{Claim}

\makeatletter
\def\GrabProofArgument[#1]{ #1: \egroup\ignorespaces}
\def\proof{\noindent\textbf\bgroup Proof%
	\@ifnextchar[{\GrabProofArgument}{. \egroup\ignorespaces}}

\makeatother

\newcommand{\alg}{{\sf GA}}

\newcommand{\ex}[1]{\operatorname{E}\left[#1\right]}

\newcommand*\samethanks[1][\value{footnote}]{\footnotemark[#1]}

\newcommand{\eat}[1]{}

\DeclareMathOperator{\OPT}{\mathrm{OPT}}
\newcommand{\dbsf}{\text{\sc degree-bounded Steiner forest}}
\newcommand{\dbst}{\text{\sc degree-bounded Steiner tree}}
\newcommand{\dbsp}{\text{\sc degree-bounded spanning tree}}
\newcommand{\odbsf}{\text{\sc online degree-bounded Steiner forest}}
\newcommand{\odbst}{\text{\sc online degree-bounded Steiner tree}}
\newcommand{\odbgst}{\text{\sc online degree-bounded group Steiner tree}}
\DeclareMathOperator{\CC}{CC}

\newcommand{\x}{\mathbf{x}}
\newcommand{\y}{\mathbf{y}}
\newcommand{\z}{\mathbf{z}}
\newcommand{\Mc}[1]{\mathcal{#1}}

\newcounter{proccnt}

\title{Online Degree-Bounded Steiner Network Design}

\author{
Sina Dehghani \thanks{University of Maryland. email: \texttt{\{dehghani, ehsani\}@umd.edu \{hajiagha, vliaghat\}@cs.umd.edu}}
\thanks{Supported in part by NSF CAREER award 1053605, NSF grant CCF-1161626, ONR YIP award N000141110662, DARPA/AFOSR grant FA9550-12-1-0423, and a Google faculty research award.}
\and Soheil Ehsani \samethanks[1] \samethanks[2]
\and MohammadTaghi Hajiaghayi \samethanks[1] \samethanks[2]
\and Vahid Liaghat \samethanks[1] \samethanks[2] \thanks{Supported in part by a Google PhD Fellowship.}
\and Harald R\"{a}cke  \thanks{Technische Universit\"{a}t M\"{u}nchen. email: raecke@in.tum.de}
}

\begin{document}

\sloppy

%
%

\date{}
\maketitle

\begin{abstract}

We initiate the study of
degree-bounded network design problems in the online setting.
The degree-bounded Steiner tree problem -- which
asks for a subgraph with minimum degree that connects a
given set of vertices -- is perhaps one of the
most representative problems in this class. This paper
deals with its well-studied generalization called the
degree-bounded Steiner forest problem where the connectivity
demands are represented by vertex pairs that need to be individually connected.
In the classical online model, the input graph is
given offline but the demand pairs arrive sequentially
in online steps. The selected subgraph starts off as the 
empty subgraph, but has to be augmented to satisfy
the new connectivity constraint in each online step.
The goal is to be competitive against an adversary that knows the input in advance.

We design a simple greedy-like algorithm that achieves a competitive
ratio of $O(\log n)$ where $n$ is the number of vertices.
We show that no (randomized) algorithm can achieve a (multiplicative)
competitive ratio $o(\log n)$; thus our result is asymptotically tight.
We further show strong hardness results for the group Steiner tree and
the edge-weighted variants of degree-bounded connectivity problems.

F{\"u}rer and Raghavachari resolved the offline variant of degree-bounded
Steiner forest in their paper in SODA'92. Since then, the natural family of
degree-bounded network design problems has been extensively studied in the
literature resulting in the development of many interesting tools and
numerous papers on the topic.
We hope that our approach in this paper, paves the way for solving the
\textit{online} variants of the classical problems in this family of network design problems.
\end{abstract}
\thispagestyle{empty}

\clearpage
\setcounter{page}{1}

\newpage

\section{Introduction}
The problem of satisfying connectivity demands on a graph while respecting
given constraints has been a pillar of the area of network design since the early
seventies~\cite{C73,CGM80,CG82,PY82}. The problem of \dbsp{}, introduced in
Garey and Johnson's \textit{Black Book} of NP-Completeness~\cite{GJ79}, was
first investigated in the pioneering work of F{\"u}rer and
Raghavachari~\cite{FR90} (Allerton'90). In the \dbsp{} problem, the goal is to
construct a spanning tree for a graph $G=(V,E)$ with $n$ vertices whose maximal
degree is the smallest among all spanning trees. Let $b^*$ denote the maximal
degree of an optimal spanning tree. F{\"u}rer and Raghavachari~\cite{FR90}
give a parallel approximation algorithm which produces a spanning tree of
degree at most $O(\log(n) b^*)$.


Agrawal, Klein, and Ravi~(\cite{AKR91a}) consider the following generalizations
of the problem. In the \dbst{} problem we are only required to connect a given
subset $T\subseteq V$. In the even more general \dbsf{} problem the demands
consist of vertex pairs, and the goal is to output a subgraph in which for
every demand there is a path connecting the pair. They design an algorithm
that obtains a multiplicative approximation factor of $O(\log(n))$.
Their main technique is to reduce the problem to
minimizing congestion under integral concurrent flow restrictions and to then
use the randomized rounding approach due to Raghavan and Thompson~(\cite{RT85},
STOC'85).

%
%
%
%


Shortly after the work of Agrawal~\textit{et al.}, in an independent work in
SODA'92 and later J.~of Algorithms'94, F{\"u}rer and
Raghavachari~\cite{FR94} significantly improved the result for \dbsf{} by
presenting an algorithm which produces a Steiner forest with maximum degree at
most $b^*+1$. They show that the same guarantee carries over to the
\textit{directed} variant of the problem as well. Their result is based on
reducing the problem to that of computing a sequence of maximal matchings on
certain auxiliary graphs. This result settles the approximability of the
problem, as computing an optimal solution is NP-hard even in the spanning tree
case.

In this paper, we initiate the study of degree-bounded network design problems
\textit{in an online setting}, where connectivity demands appear over time and must be
immediately satisfied. We first design a deterministic algorithm for \odbsf{}
with a logarithmic competitive ratio. Then we show that this competitive ratio
is asymptotically best possible by proving a matching lower bound for
randomized algorithms that already holds for the Steiner tree variant of the problem.


In the offline scenario, the results of F{\"u}rer,
Raghavachari~\cite{FR90,FR94} and Agrawal, Klein, Ravi~\cite{AKR91a} were the
starting point of a very popular line of work on various degree-bounded network
design problems~\cite{MRSR98,G06,N12,LS13,KKN13,EV14}. We refer the reader to
the next sections for a brief summary. One particular variant that has been
extensively studied is the \textit{edge-weighted} \dbsp{}. Initiated by
Ravi~\textit{et al.} (\cite{MRSR98}, J.~of Algorithms'98), in this version, we
are given a weight function over the edges and a bound $b$ on the maximum
degree of a vertex. The goal is to find a minimum-weight spanning tree with
maximum degree at most $b$. The groundbreaking results obtained by
Goemans~(\cite{G06}, FOCS'06) and Singh and Lau~(\cite{LS07},~STOC'07) settle
the problem by giving an algorithm that computes a minimum-weight spanning tree
with degree at most $b+1$. Slightly worse result are obtained by Singh and
Lau~(\cite{LS13}, STOC'08) for the Steiner tree variant.
Unfortunately, in the online setting it is not possible to obtain a comparable
result. We show that for any (randomized) algorithm $\Mc{A}$ there exists a request
sequence such that $\Mc{A}$ outputs a sub-graph that either has weight
$\Omega(n)\cdot\mathrm{OPT}_b$ or maximum degree $\Omega(n)\cdot b$.

%



\subsection{Our Contributions}
In the online variant of \dbsf{}, we are given the graph $G$ in advance,
however, demands arrive in an online fashion. At online step $i$, a new
demand $(s_i,t_i)$ arrives. Starting from an empty subgraph, at each step the
online algorithm should augment its solution so that the endpoints of the new
demand $s_i$ and $t_i$ are connected. The goal is to minimize the maximum
degree of the solution subgraph. In the \textit{non-uniform} variant of the
problem, a degree bound $b_v\in \mathbb{R}^+$ is given for every vertex $v$.
For a subgraph $H$ and a vertex $v$, let $\deg_H(v)$ denote the degree of $v$
in $H$. The \textit{load} of a vertex is defined as the ratio
${\deg_H(v)}/{b_v}$. In the non-uniform variant of \odbsf{}, the goal is to
find a subgraph satisfying the demands while minimizing the maximum load of a
vertex.

Our algorithm is a simple and intuitive greedy algorithm. Upon the arrival of a
new demand $(s_i,t_i)$, the greedy algorithm (GA) satisfies the demand by
choosing an $(s_i,t_i)$-path $P_i$ such that after augmenting the solution with
$P_i$, the maximum load of a vertex in $P_i$ is minimum. 
A main result of our
paper is to prove that the maximum load of a vertex in the output of GA is
within a logarithmic factor of $\OPT$, the maximum load of a vertex in an
optimal offline solution which knows all the demands in advance.

\begin{theorem}\label{thm:mainlog}
	The algorithm GA produces an output with maximum load at most
	$O(\log n)\cdot \OPT$.
\end{theorem}

The crux of our analysis is establishing several structural properties of the
solution subgraph. First we group the demands according to the maximum load of
the bottleneck vertex at the time of arrival of the demand. We then show that
for every threshold $r>0$, vertices with load at least $r$ at the end of the
run of GA, form a cut set that well separates the group of demands with load at
least $r$ at their bottleneck vertex. Since the threshold value can be chosen
arbitrarily, this leads to a series of cuts that form a chain. The greedy
nature of the algorithm indicates that each cut highly disconnects the demands.
Intuitively, a cut that highly disconnects the graph (or the demands) implies a
lower bound on the number of branches of every feasible solution.

We use a natural dual-fitting argument to show that for every cut set, the
ratio between the number of demands in the corresponding group, over the total
degree bound of the cut, is a lower bound for $\OPT$. Hence, the problem comes
down to showing that one of the cuts in the series has ratio at least $1/O(\log
n)$ fraction of the maximum load $h$ of the output of GA. To this end, we first
partition the range of $r\in (0,h]$ into $O(\log n)$ layers based on the total
degree bound of the corresponding cut. We then show that the required cut can
be found in an interval with maximum range of $r$. We analyze GA formally in
Section~\ref{sec:sfalg}.

We complement our first theorem by giving an example for a special case of
\odbst{} in which no online (randomized) algorithm can achieve a competitive
ratio better than $\Omega(\log n)$.

\begin{theorem}\label{thm:LBst}
	Any (randomized) online algorithm for the degree bounded online Steiner tree
	problem has (multiplicative) competitive ratio $\Omega(\log n)$. This already holds when $b_v=1$
	for every node.
\end{theorem}

We further investigate the following extensions of the online degree bounded
Steiner tree problem. 
First, we
consider the edge-weighted variant of the degree-bounded Steiner tree problem.
Second, we consider the group Steiner tree version in
which each demand consists of a subset of vertices, and the goal is to
find a tree that covers at least one vertex of each demand group. 
The following theorems show that one cannot obtain a non-trivial competitive
ratio for these versions in their general form.\footnote{Our lower bound results imply that one needs
  to restrict the input in order to achieve competitiveness. In particular for
  the edge-weighted variant, our proof does not rule out the existence of a
  competitive algorithm when the edge weights are polynomially bounded. 
}

%
%

\begin{theorem}\label{thm:ewLB}
Consider the edge weighted variant of \odbst{}. For any (randomized) online
algorithm $\Mc{A}$, there exists an instance and a request sequence such that 
either $\ex{\operatorname{maxdegree}(\Mc{A})}\ge\Omega(n)\cdot b$ or
$\ex{\operatorname{weight}(\Mc{A})}\ge\Omega(n)\cdot \OPT_b$, where $\OPT_b$ denotes
the minimum weight of a Steiner tree with maximum degree $b$.
%
\end{theorem}

\begin{theorem}\label{thm:gstLB}
  There is no deterministic algorithm with competitive ratio $o(n)$ for the
  {\sc degree-bounded group Steiner tree} problem.
\end{theorem}

\subsection{Related Degree-Bounded Connectivity Problems}\label{sec:related}
The classical family of degree-bounded network design problems have various
applications in broadcasting information, package distribution, decentralized
communication networks, etc.~(see e.g.\
\cite{garcia1988implementation,huang2003publish}). Ravi~\textit{et
  al.}~(\cite{MRSR98},~J. Algorithms'98), first considered the general
\textit{edge-weighted} variant of the problem. They give a bi-criteria $(O(\log
n), O(\log n)\cdot b)$-approximation algorithm, i.e., the degree of every node
in the output tree is $O(\log n)\cdot b$ while its total weight is $O(\log
n)$ times the optimal weight. A long line of work~(see e.g.
\cite{KR00},~STOC'00 and \cite{KR05},~SIAM J. C.) was focused on this problem
until a groundbreaking breakthrough was obtained by
Goemans~(\cite{G06},~FOCS'06); his algorithm computes a minimum-weight spanning
tree with degree at most $b+2$. Later on, Singh and Lau~(\cite{LS07},~STOC'07)
improved the degree approximation factor by designing an algorithm that outputs
a tree with optimal cost while the maximum degree is at most $b+1$.

In the \textit{degree-bounded survivable network design} problem, a number
$d_i$ is associated with each demand $(s_i, t_i)$. The solution subgraph should
contain at least $d_i$ edge-disjoint paths between $s_i$ and $t_i$. Indeed this
problem has been shown to admit bi-criteria approximation algorithms with
constant approximation factors~(e.g. \cite{LS13},~STOC'08). We refer the reader to a recent survey in
\cite{lau2011iterative}. This problem has been recently considered in the
node-weighted variant too (see e.g. \cite{N12,EV14}). The degree-bounded
variant of several other problems such as $k$-MST and $k$-arborescence has also
been considered in the offline setting for which we refer the reader to
\cite{KKN13,BKN09} and references therein.

\subsection{Related Online Problems}
Online network design problems have attracted substantial attention in the last
decades. The online edge-weighted Steiner tree problem, in which the goal is to
find a minimum-weight subgraph connecting the demand nodes, was first
considered by Imase and Waxman~(\cite{IW91},~SIAM J. D. M.'91). They showed
that a natural greedy algorithm has a competitive ratio of $O(\log n)$, which
is optimal up to constants. This result was generalized to the online
edge-weighted Steiner forest problem by Awerbuch \textit{et
  al.}~(\cite{AAB96},~SODA'96) and Berman and Coulston~(\cite{BC97},~STOC'97).
Later on, Naor, Panigrahi, and Singh~(\cite{NPS11}),~FOCS'11) and Hajiaghayi,
Liaghat, and Panigrahi~(\cite{HLP13},~FOCS'13), designed poly-logarithmic
competitive algorithms for the more general \textit{node-weighted} variant of
Steiner connectivity problems. This line of work has been further investigated
in the prize-collecting version of the problem, in which one can ignore a
demand by paying its given penalty. Qian and Williamson~(\cite{QW11},~ICALP'11)
and Hajiaghayi, Liaghat, and Panigrahi~(\cite{HLP14},~ICALP'14) develop
algorithms with a poly-logarithmic competitive algorithms for these variants.

The online $b$-matching problem is another related problem in which
vertices have degree bounds but the objective is to maximize the size of the
solution subgraph. In the worst case model, the celebrated result of Karp
\textit{et al.}~(\cite{KVV90},~STOC'90) gives a $(1-1/e)$-competitive
algorithm. Different variants of this problem have been extensively studied in
the past decade, e.g., for the random arrival model see
\cite{FMMM09,KMT11,MY11,MSVV07}, for the full information model see
\cite{MOS12,MOZ12}, and for the prophet-inequality model see
\cite{AHL11,AHL12,AHL13}. We also refer the reader to the comprehensive
survey by Mehta~\cite{M12}.

Many of the problems mentioned above can be described with an online packing or
covering linear program. Initiated by work of Alon \textit{et
  al.}~\cite{AAABN09} on the online set cover problem, Buchbinder and Naor
developed a strong framework for solving packing/covering LPs fractionally
online. For the applications of their general framework in solving numerous
online problems, we refer the reader to the survey in \cite{BN09}. Azar
\textit{et al.}~\cite{ABFP13} generalize this method for the fractional
\textit{mixed} packing and covering LPs. In particular, they show an
application of their method for integrally solving a generalization of
capacitated set cover. Their result is a bi-criteria competitive algorithm that
violates the capacities by at most an $O(\log^2 n)$ factor while the cost of
the ouput is within $O(\log^2 n)$ factor of optimal.
We note that although the fractional variant of our problem is a special case of online mixed packing/covering LPs, we do not know of any online rounding method for Steiner connectivity problems.

\subsection{Preliminaries}
Let $G=(V, E)$ denote an undirected graph of size $n$ ($|V|=n$). For every
vertex $v\in V$, let $b_v\in \mathbb{R}^+$ denote the \textit{degree bound} of
$v$. In the \dbsf{} problem, we are given a sequence of connectivity
\textit{demands}. The $i^{th}$ demand is a pair of vertices $(s_i,t_i)$ which
we call the \textit{endpoints} of the demand. An algorithm outputs a subgraph
$H\subseteq G$ in which for every demand its endpoints are connected. The
\textit{load} of a vertex $v$ w.r.t. $H$ is defined as
$\ell_H(v)={\deg_H(v)}/{b_v}$. We may drop the subscript $H$ when it is clear
from the context. The goal is to find a subgraph $H$ that minimizes the maximum
load of a vertex. Observe that one can always find an optimal solution without
a cycle, hence the name Steiner \textit{forest}. Furthermore, without loss of
generality\footnote{One can add a dummy vertex for every vertex $v\in V$
  connected to $v$. We then interpret a demand between two vertices as a demand
  between the corresponding dummy vertices. The degree bound of a dummy vertex
  can be set to infinity. This transformation can increase the degree of any
  node by at most a factor of $2$, which does not affect our asymptotic
  results.}, we assume that the endpoints of the demands are distinct vertices
with degree one in $G$ and degree bound infinity. We denote the maximum load of
a vertex in an optimal subgraph by $\OPT=\min_{H} \max_{v} \ell_H(v)$.

The following mixed packing/covering linear program (\ref{LP:sf}) is a
relaxation for the natural integer program for \dbsf{}. Let $\Mc{S}$ denote the
collection of subsets of vertices that separate the endpoints of at least one
demand. For a set of vertices $S$, let $\delta(S)$ denote the set of edges with
exactly one endpoint in $S$. In \ref{LP:sf}, for an edge $e$, $\x(e)=1$
indicates that we include $e$ in the solution while $\x(e)=0$ indicates
otherwise. The variable $\alpha$ indicates an upper bound on the load of every
vertex. The first set of constraints ensures that the endpoints of every demand
is connected. The second set of constraints ensures that the load of a vertex is
upper bounded by $\alpha$. The program \ref{LP:dualsf} is the dual of the LP
relaxation \ref{LP:sf}.

\def\myrule{\phantom{\rule[-3.2ex]{1pt}{5.6ex}}}
\noindent
\begin{minipage}{.44\textwidth} %
\begin{align*}
\text{minimize}           &\quad \myrule\alpha \tag{{$\mathbb{P}$}}\label{LP:sf}\\
\forall S\subseteq \Mc{S} &\quad \myrule\zsum{e\in \delta\smash{(S)}}{\x(e)} \geq 1 \tag{$\y(S)$}\\
\forall v\in V            &\quad \myrule\zsum{e\in \delta\smash{(\{v\})}}{\x(e)} \leq \alpha \cdot b_v \tag{$\z(v)$} \\
                          &\quad \myrule\x(e), \alpha\in \mathbb{R}_{\geq 0} 
\end{align*}
\end{minipage} %
\hspace{0.02\textwidth}
\begin{minipage}{.53\textwidth} %
\begin{align*}
\text{maximize} &\quad  \myrule\zsum{S\in \Mc{S}}{\y(S)} \tag{{$\mathbb{D}$}}\label{LP:dualsf}\\
\forall e=(u,v) &\quad  \myrule\zsum{S: e\in \delta\smash{(S)}}{\y(S)} \leq \z(u)+\z(v) \tag{$\x(e)$}\\
                &\quad  \myrule\zsum{v}{\z(v)b_v} \leq 1 \tag{$\alpha$} \\
                &\quad  \myrule\z(v), \y(S)\in \mathbb{R}_{\geq 0}
\end{align*}
\end{minipage}
In the online variant of the problem, $G$ and the degree bounds are known in
advance, however, the demands are given one by one. Upon the arrival of the
$i^{th}$ demand, the online algorithm needs to output a subgraph $H_i$ that
satisfies all the demands seen so far. The output subgraph can only be
augmented, i.e., for every $j<i$, $H_j\subseteq H_i$. The \textit{competitive
  ratio} of an online algorithm is then defined as the worst case ratio of
${\max_v \ell_{H}(v)}/{\OPT}$ over all possible demand sequences where
$H$ is the final output of the algorithm.


\section{Online Degree-Bounded Steiner Forest}\label{sec:sfalg}
Consider an arbitrary online step where a new demand $(s,t)$ arrived. Let
$H$ denote the online output of the previous step.
In order to augment $H$ for connecting $s$ and $t$, one can shortcut through
the connected components of $H$. We say an edge $e=(u,v)$ is an
\textit{extension edge} w.r.t. $H$, if $u$ and $v$ are not connected in $H$. Let
$P$ denote an $(s,t)$-path in $G$. The extension part $P^*$ of $P$ is defined
as the set of extension edges of $P$. Observe that augmenting $H$ with $P^*$
satisfies the demand $(s,t)$. For a vertex $v$, we define
$\ell^+_H(v)=\ell_H(v)+{2}/{b_v}$ to be the \textit{uptick load of $v$}. 
We slightly misuse the notation by defining $\ell^+_H(P^*)=\max_{v \in
  V(P^*)} \ell^+_H(v)$ as the uptick load of $P^*$, where $V(P^*)$ is the set of
endpoints of edges in $P^*$.

The \textit{greedy algorithm} (GA) starts with an empty subset $H$. Upon 
arrival of the $i$-{th} demand $(s_i,t_i)$, we find a path $P_i$ with 
smallest uptick load $\ell^+_H(P_i^*)$ where $P_i^*$ is the extension part of
$P_i$. We break ties in favor of the path with a smaller number of edges.
Note that the path $P_i$ can be found efficiently using Dijkstra-like
algorithms. We define $\tau_i=\ell^+_H(P_i^*)$ to be the \textit{arrival threshold}
of the $i$-{th} demand. We satisfy the new demand by adding $P_i^*$ to the edge
set of $H$ (see Algorithm~\ref{alg:dbsf}).

\subsection{Analysis}
We now use a dual-fitting approach to show that GA has an
asymptotically tight competitive ratio of $O(\log(n))$. 
In the following we  use \alg{} to also refer to the \emph{final output}
of our greedy algorithm.
We first show several structural properties of
\alg{}. We then use these combinatorial properties to construct a family of
feasible dual vectors for \ref{LP:dualsf}. We finally show that there always
exists a member of this family with an objective value of at least a
${1}/{O(\log(n))}$ fraction of the maximum load of a vertex in \alg{}. This
in turn proves the desired competitive ratio by using weak duality.

For a real value $r\geq 0$, let $\Gamma(r)$ denote the set of vertices with
$\ell^+_{\alg}(v)\geq r$. Let $D(r)$ denote the indices of demands $i$ for
which the arrival threshold $\tau_i$ is at least $r$. For a subgraph $X$, let
$\CC(X)$ denote the collection of connected components of $X$. For ease of
notation, we may use the name of a subgraph to denote the set of vertices of
that subgraph as well. Furthermore, for a graph $X$ and a subgraph $Y\subseteq
X$, let $X\setminus Y$ denote the graph obtained from $X$ by removing the
vertices of $Y$.

Recall that $\Mc{S}$ is the collection of subsets that separate the endpoints
of at least one demand. The following lemma, intuitively speaking, implies that
$\Gamma(r)$ well-separates the endpoints of $D(r)$.

\begin{lemma}\label{lem:multiwaycut}
For any threshold $r> 0$, we have
$\left| \CC(G\setminus\Gamma(r)) \cap \Mc{S}  \right| \geq |D(r)|+1 $.
\end{lemma}
\begin{proof}
  For a vertex $v\in G\setminus\Gamma(r)$ we use $\operatorname{CC}(v)$ to
  denote its connected component. Observe that, since the endpoints of demands
  are nodes with infinite degree bound, the endpoints are \emph{not} contained in $\Gamma(r)$, and, hence, are
  in $G\setminus\Gamma(r)$.

  We construct a graph $F$ that has one node for every connected component in
  $G\setminus\Gamma(r)$ that contains an endpoint of a demand in $D(r)$. Edges
  in $F$ are defined as follows. For every demand $i\in D(r)$ between $s_i$ and
  $t_i$, we add an edge that connects the components $\operatorname{CC}(s_i)$
  and $\operatorname{CC}(t_i)$. In the following we argue that $F$ does not
  contain cycles. This gives the lemma since in a forest $|D(r)|+1=|E_F|+1\le|V_F|\le 
  |\CC(G\setminus\Gamma(r)) \cap \Mc{S}|$.

  Assume for contradiction that the sequence $(e_{i_0},\dots,e_{i_{k-1}})$,
  $k\ge 2$ forms a (minimal) cycle in $F$, where $e_{i_j}$ corresponds to the
  demand between vertices $s_{i_j}$ and $t_{i_j}$ (see Figure~\ref{fig:lemsep}
  in the Appendix). Assume wlog.\ that $e_{i_{k-1}}$ is the edge of the cycle for
  which the corresponding demand appears last, i.e., $i_{k-1}\geq i_j$ for
  every $j<k$. Let $H$ denote the online solution at the time of arrival of the
  demand $i_{k-1}$. We can augment $H$ to connect each $t_{i_j}$ to
  $s_{i_{j+1\bmod k}}$, $0\le j\le k-1$ without using any node from
  $\Gamma(r)$, since these nodes are in the same component in
  $G\setminus\Gamma(r)$. But then we have a path $P$ between $s_{i_{k-1}}$ and
  $t_{i_{k-1}}$ and the extension part $P^*$ does not contain any vertex from
  $\Gamma(r)$. This is a contradiction since the arrival threshold for the
  demand $i_{k-1}$ is at least $r$.
\end{proof}

Lemma~\ref{lem:multiwaycut} shows that cutting $\Gamma(r)$ highly disconnects
the demands in $D(r)$. Indeed this implies a bound on the \textit{toughness} of
the graph. Toughness, first defined by Chv{\'a}tal~\cite{C73} and later
generalized by Agrawal \textit{et al.}~\cite{AKR91b}, is a measure of how easy
it is to disconnect a graph into many components by removing a few vertices.
More formally, the toughness of a graph is defined as $\min_{X\subseteq V}
\frac{|X|}{|\CC(G\setminus X)|}$. For the spanning tree variant of the problem,
it is not hard to see that $\OPT$ is at least the reciprocal of toughness.
Although it is more involved, we can still establish a similar relation for the
non-uniform Steiner forest problem (see Lemma~\ref{lem:fitting}). However,
first we need to show a lower bound on the number of demands separated by
$\Gamma(r)$.


We establish a lower bound for $|D(r)|$ with respect to the load of vertices in
$\Gamma(r)$.
For any $r,b> 0$ we define $\Gamma_b(r):=\{\ell_{\alg}^+(v)\ge r\wedge b_v\ge
b\}$, as the set of nodes with degree bound at least $b$ that have uptick load at least
$r$ in the final online solution. We further define
\begin{equation*}
\operatorname{excess}(r,b)=\sum_{v\in
  \Gamma_b(r)}\!\!\Big(\operatorname{deg}_{\alg}(v)-\lceil r b_v\rceil+2\Big)\enspace,
\end{equation*}
which sums the (absolute) load that nodes in $\Gamma_b(r)$ receive \emph{after}
their \emph{uptick load} became $r$ or larger. The following lemma relates $|D(r)|$ to 
$\operatorname{excess}(r,b)$.
\begin{lemma} \label{lem:largedemand}
For any $r,b>0$, 
$\operatorname{excess}(r,b)\le 2|D(r)|+3|\Gamma_b(v)|$.
\end{lemma}
\begin{proof}
  Consider some online step $i$.
  Let $H$ denote the output of the previous step. Let $P^*_i$ be the extension
  part of the path selected by GA and let $V(P^*_i)$ denote the endpoints of
  edges of $P^*_i$. Since in GA we break ties in favor of the path with the
  smaller number of edges, we can assume that the path does not go through a
  connected component of $H$ more than once, i.e., for every $C\in \CC(H)$,
  $|V(P^*_i)\cap C| \leq 2$.

  Consider the variable quantity $\delta(r,b):=\sum_{v: \ell^+_H(v)\geq r
    \wedge b_v\geq b}(\deg_H(v)-\lceil r b_v \rceil+2)$ throughout the steps of GA
  where $H$ denotes the output of the algorithm at every step. Intuitively, at
  each step $\delta(r,b)$ denotes the total increment in degrees of those
  vertices in $\Gamma_b(r)$ that already reached uptick load $r$. In
  particular, at the end of the run of GA,
  $\delta(r,b)=\operatorname{excess}(r,b)$.

  Now suppose at step $i$ adding the edges $P^*_i$ to $H$ increases
  $\delta(r,b)$ by $q_i$. There are two reasons for such an increase. On the
  one hand, if the demand $i$ is from $D(r)$ it may simply increase 
  $\operatorname{deg}_H(v)$ for some node with uptick load at least $r$. On the
  other hand also if the demand is not from $D(r)$ it may cause a node to
  increase its uptick load to $r$ in which case it could contribute to the
  above sum with at most 1 (in a step the degree may increase by 2; the first
  increase by 1 could get the uptick load to $\ge r$ and the second increase
  contributes to the sum). 

  Clearly increases of the second type can accumulate to at most
  $|\Gamma_b(r)|$. In order to derive a bound on the first type of increases
  recall that we assume that endpoints of demands are distinct vertices with
  degree one and thus $s_i$ and $t_i$ are outside any connected component of
  $H$. Since $V(P^*_i)$ contains at most two vertices of a connected component
  of $H$, we can assert that the path selected by GA is passing through at
  least $q_i/2$ components of $H$ that contain some vertices of $\Gamma_b(r)$.
  Hence, after adding $P^*_i$ to the solution, the number of connected
  components of the solution that contain some vertices of $\Gamma_b(r)$
  decreases by at least $(q_i-2)/2$. However, throughout all the steps, the
  total amount of such decrements cannot be more than the number of vertices in
  $\Gamma_b(r)$.
%
%
%
Therefore
\begin{align*}
\operatorname{excess}(r,b)
	&=\sum_{i\in D(r)} q_i+\sum_{i\notin D(r)} q_i\\
	&= 2 |D(r)| + \sum_{i\in D(r)} (q_i-2) +|\Gamma_b(r)|\\
    &\le 2|D(r)| +3|\Gamma_b(r)|\enspace,
\end{align*}
and the lemma follows.
\end{proof}

Let $\Delta>0$ denote the minimum degree bound of a vertex with non-zero degree in
an optimal solution. Note that this definition implies that $\OPT\ge1/\Delta$. For a set of vertices $\Gamma$, let $b(\Gamma)=\sum_{v\in \Gamma} b_v$.
We now describe the main dual-fitting argument for
bounding the maximum load in \alg{}. 

\begin{lemma} \label{lem:fitting}
	For any $r>0$, ${|D(r)|}/{b(\Gamma_{\Delta}(r))}\leq \OPT$.
\end{lemma}
\begin{proof}
  Let $G_{\Delta}$ denote the graph obtained from $G$ by removing vertices with
  degree bound below $\Delta$. Similarly, let $\Mc{S}_{\Delta}$ denote the
  collection of sets that separate a demand in $G_{\Delta}$. Consider a
  slightly modified dual program $\text{\ref{LP:dualsf}}_{\Delta}$ defined on
  $G_{\Delta}$ and $\Mc{S}_{\Delta}$, i.e., we obtain
  $\text{\ref{LP:dualsf}}_{\Delta}$ from the original dual program by removing
  all vertices with degree bound below $\Delta$. We note that the objective
  value of a dual feasible solution for $\text{\ref{LP:dualsf}}_{\Delta}$ is
  still a lower bound for $\OPT$. In the remainder of the proof, we may refer
  to $\text{\ref{LP:dualsf}}_{\Delta}$ as the dual program.
  Recall that in a feasible dual solution, we need to define a dual value
  $\y(S)$ for every cut $S\in \Mc{S}_{\Delta}$ such that for every edge
  $e=(u,v)$ the total value of cuts that contain $e$ does not exceed
  $\z(u)+\z(v)$. On the other hand, $\sum_v \z(v) b_v$ cannot be more than one.

  For a
  real value $r> 0$, we construct a dual vector $(\y_r,\z_r)$ as follows.
  For every $v\in \Gamma_{\Delta}(r)$, set
  $\z_r(v)={1}/{b(\Gamma_{\Delta}(r))}$; for other vertices set
  $\z_r(v)=0$. For every component $S\in \CC(G_{\Delta}\setminus
  \Gamma_{\Delta}(r)) \cap \Mc{S}_{\Delta}$, set
  $\y_r(S)={1}/{b(\Gamma_{\Delta}(r))}$; for all other sets set
  $\y_r(S)=0$. We prove the lemma by showing the feasibility of the dual vector
  $(\y_r,\z_r)$ for $\text{\ref{LP:dualsf}}_{\Delta}$.

  Consider an arbitrary component $C\in \CC(G\setminus \Gamma(r))$. By
  definition, $C$ separates at least one demand. Let $i$ be such a demand and
  let $t_i\in C$ denote an endpoint of it. Removing vertices with degree bound
  below $1/\Delta$ from $C$ may break $C$ into multiple smaller components.
  However, an endpoint of a demand has degree bound infinity, and, hence, the
  component that contains $t_i$ belongs to $\Mc{S}_\Delta$. Therefore
  $|\CC(G\setminus \Gamma(r)) \cap \Mc{S}|\leq |\CC(G_{\Delta}\setminus
  \Gamma_{\Delta}(r)) \cap \Mc{S}_{\Delta}|$; which by
  Lemma~\ref{lem:multiwaycut} leads to $|\CC(G_{\Delta}\setminus
  \Gamma_{\Delta}(r)) \cap \Mc{S}_{\Delta}| \geq |D(r)|+1$. Therefore the dual
  objective for the vector $(\y_r,\z_r)$ is at least

  \[ \sum_{S\in \Mc{S}_{\Delta}} \y_r(S) = \sum_{S\in \CC(G_{\Delta}\setminus
    \Gamma_{\Delta}(r)) \cap \Mc{S}_{\Delta}} \frac{1}{b(\Gamma_{\Delta}(r))}
  \geq \frac{|D(r)|+1}{b(\Gamma_{\Delta}(r))} \enspace .\]
  Thus we only need to show that $(\y_r,\z_r)$ is feasible for
  Program~$\text{\ref{LP:dualsf}}_{\Delta}$. First, by construction we have
  \[ \sum_v \z_r(v) b_v = \sum_{v\in \Gamma_{\Delta}(r)}
  \frac{1}{b(\Gamma_{\Delta}(r))} \cdot b_v =1 \enspace . \]
  Now consider an arbitrary edge $e=(u,v)$. If $e$ does not cross any of the
  components in $\CC(G_{\Delta}\setminus \Gamma_{\Delta}(r))$, then $\sum_{S:
    e\in \delta(S)} \y_r(S)=0$ and we are done. Otherwise, $\sum_{S: e\in
    \delta(S)} \y_r(S)={1}/{b(\Gamma_{\Delta}(r))}$. However, exactly one
  endpoint of $e$ is in $\Gamma_{\Delta}(r)$. Thus
  $\z_r(u)+\z_r(v)={1}/{b(\Gamma_{\Delta}(r))}$, which implies that the
  dual vector is feasible.
\end{proof}

\noindent
We now have all the ingredients to prove the main theorem.
\smallskip

\begin{proof}[of Theorem~\ref{thm:mainlog}]
  Let \alg{} denote the final output of the greedy algorithm. Let $h$ denote
  the maximum load of a vertex in $\alg$, i.e., $h=\max_v \ell_{\alg}(v)$.
  Furthermore, let $h_{\Delta}=\max_{v: b_v\geq \Delta} \ell_{\alg}(v)$. 
In the following we first show that $h_{\Delta}\le O(\log n)\cdot\OPT$.
For this we use the folowing claim.


\begin{claim}
There exists an $r>0$ such that 
\begin{equation*}
\operatorname{excess}(r,\Delta)\ge
\frac{h_\Delta-1/\Delta}{4\log_2n+6}\cdot b(\Gamma_\Delta(r))-|\Gamma_\Delta(r)|\enspace.
\end{equation*}
\end{claim}
\begin{proof}
  Recall that $\Gamma_{\Delta}(r)$ is the set of vertices with uptick load at
  least $r$ in \alg{} and degree bound at least $\Delta$.
  The function $b(\Gamma_\Delta(r))$ is non-increasing with $r$ since for every
  $r_1<r_2$, $\Gamma_{\Delta}(r_2)\subseteq \Gamma_{\Delta}(r_1)$. For
  $r=1/\Delta$, $b(\Gamma_\Delta(r))\le n(n+1)\Delta$ as a node with  degree bound
  $b_v> (n+1)\Delta$ will have uptick load at most $(n+1)/b_v<{1}/{\Delta}$,
  and, hence, will not be in $\Gamma_\Delta(r)$.
  Also, $r<h_\Delta$ implies that $b(\Gamma_\Delta(r))\ge\Delta$ as there
  exists a node with load $h_\Delta$ in $\Gamma_\Delta(r)$ and this node has
  degree bound at least $\Delta$.

  We now partition the range of $r$ (from $1/\Delta$ to $h_\Delta$) into
  logarithmically many intervals. We define the $q$-th interval by
\begin{align*}
U(q) &= \left\{ r\mid {1}/{\Delta} 
\le r < h_{\Delta} \ \wedge \ 2^q/\Delta \leq b(\Gamma_{\Delta}(r)) <
2^{q+1}/\Delta \right\} \enspace,
\end{align*}
for $q\in\{0,\dots,\lceil\log_2((n+1)n)\rceil\}$.
We further set
$\overline{r}(q) = \max U(q)$ and $\underline{r}(q) = \min U(q)$,
and call $\overline{r}(q)-\underline{r}(q)$ the length of the $q$-th interval.
Since there are only $2\log_2n+3$ possible choices for $q$ there must exist an
interval of length at least $(h_\Delta-1/\Delta)/(2\log_2n+3)$.

Consider a node $v$ that is contained in $\Gamma_\Delta(\overline{r})$
and, hence, also in $\Gamma_\Delta(\underline{r})$. This node starts
contributing to $\operatorname{excess}(\underline{r},\Delta)$, once its uptick
load reached $\underline{r}$ and contributes at least until its uptick load
reaches $\overline{r}$. Hence, the total contribution is at least
$\operatorname{deg}(\overline{r})-\operatorname{deg}(\underline{r})$, where
$\operatorname{deg}(\overline{r})$ and $\operatorname{deg}(\underline{r})$
denotes the degree of node $v$ when reaching uptick load $\overline{r}$ and
$\underline{r}$, respectively. We have,
\begin{equation*}
\begin{split}
\operatorname{deg}(\overline{r})-\operatorname{deg}(\underline{r})
&= (\lceil\overline{r}b_v\rceil-2)-(\lceil\underline{r}b_v\rceil-2)\\
&\ge (\overline{r}-\underline{r})b_v-1\enspace.
\end{split}
\end{equation*}
Summing this over all nodes in $\Gamma_\Delta(\overline{r})$ 
gives 
\begin{equation*}
\begin{split}
\operatorname{excess}(\underline{r},\Delta) 
  &\ge (\overline{r}-\underline{r})\cdot b(\Gamma_\Delta(\overline{r}))-|\Gamma_\Delta(\overline{r})|\\
  &\ge \frac{1}{2}(\overline{r}-\underline{r})\cdot b(\Gamma_\Delta(\underline{r}))-|\Gamma_\Delta(\underline{r})|\\
  &\ge \frac{h_\Delta-1/\Delta}{4\log_2n+6}\cdot b(\Gamma_\Delta(\underline{r}))-|\Gamma_\Delta(\underline{r})|\enspace,
\end{split}
\end{equation*}
where the second inequality uses the fact that
$|\Gamma_\Delta(\overline{r})|\le|\Gamma_\Delta(\underline{r})|\le|\Gamma_\Delta(\overline{r})|/2$.
\end{proof} 

\begin{claim}
$h_{\Delta}\leq (24\log_2 n+37) \OPT$.
\end{claim}
\begin{proof}
If we use the $r$ from the previous claim and solve for $h_\Delta$ we obtain
\begin{equation*}
\begin{split}
h_\Delta
&\le(4\log_2 n+6)\cdot\frac{\operatorname{excess}(r,\Delta)+\Gamma_\Delta(r)}{b(\Gamma_\Delta(r))}+\frac{1}{\Delta}\\
&\le(4\log_2 n+6)\cdot\frac{2D(r)+4\Gamma_\Delta(r)}{b(\Gamma_\Delta(r))}+\frac{1}{\Delta}\\
&\le(4\log_2 n+6)\cdot \Big(2\OPT+4\frac{1}{\Delta}\Big)+\frac{1}{\Delta}\\
&\le(24\log_2 n+37)\cdot \OPT\enspace.
\end{split}
\end{equation*}
Here the second inequality is due to Lemma~\ref{lem:largedemand}, the third
uses Lemma~\ref{lem:fitting} and the fact that
$b(\Gamma_\Delta(r))\ge\Delta|\Gamma_\Delta(r)|$. The last inequality uses $\OPT\ge 1/\Delta$.
\end{proof}

\medskip
\noindent
We now bound the maximum load $h$ in terms of the restricted maximum load $h_{\Delta}$.
Consider a vertex $v^*$ with maximum load $\ell_{\alg}(v^*)=h$. If $b_{v^*}\geq
\Delta$ then $h_{\Delta}=h$ and we are done. Otherwise consider the last
iteration $i$ in which the degree of $v^*$ is increased in the solution. Let
$H$ denote the output of the algorithm at the end of the previous iteration $i-1$. The
degree of $v^*$ in the online solution is increased by at most two at iteration
$i$. Hence
\[h	\leq \ell^+_H(v^*) \leq \tau_i  \enspace .\]
Recall that in our greedy algorithm, $\tau_i$ is the minimum uptick load of a
path that satisfies the new demand. Let $P$ denote a path that connects $s_i$
and $t_i$ in an optimal solution. Recall that by the definition, $b_v\geq
\Delta$ for every vertex $v$ of $P$. For every vertex $v$ in $P$ we have
\begin{align*}
\tau_i &\leq \ell^+_H(v) \\
	&\leq h_{\Delta}+\frac{2}{b_v} &\ell_H(v)\leq \ell_{\alg}(v)\leq h_{\Delta}\\
	&\leq h_{\Delta}+\frac{2}{\Delta} & b_v\geq \Delta\\
	&\leq h_{\Delta}+2\OPT	 & \OPT\geq 1/\Delta\\
	&\leq O(\log n)\cdot\OPT &\text{by the above claim}
\end{align*}
Therefore leading to $h \leq O(\log n)\cdot\OPT$.
\end{proof}

\def\mydeg{\operatorname{deg}}
\def\mydegb{\operatorname{deg}'}

\section{An Asymptotically Tight Lower Bound}
In the following we show a lower bound for 
\odbst{}.
Consider a  graph $G=(X \uplus Z, E)$, with $|Z|=2^\ell$ and
$|X|=\binom{2^\ell}{2}$. 
For every pair $\{z_1,z_2\}$ of nodes from $Z$ there exists 
an edge $\{z_1,z_2\}\in E$ and a node $x\in X$ that is connected
to $z_1$ and $z_2$. An arbitrary node from $Z$ acts as root for the  Steiner
tree problem.

For a node $z\in Z$, an algorithm $A$, and a request sequence $\sigma$
(consisting of nodes from $X$) we use $\mydeg_{A,\sigma}(z)$ to denote the
number of neighbors of $z$ \emph{among all nodes from X} in the Steiner tree
obtained when running algorithm $A$ on request sequence $\sigma$. Similarly, we
use $\mydegb_{A,\sigma}(z)$ to denote the number of those neighbors of $z$ that
also appear in $\sigma$. Note that $\mydeg(\cdot)$ and $\mydegb(\cdot)$ ignore
edges between nodes from $Z$ as an algorithm (online or offline) can simply
connect all nodes in $Z$ in a cycle which increases the degree of any node by
only 2.

\def\ex{\operatorname{E}}
\begin{lemma} Fix a possibly randomized online algorithm $A$.
For any subset $S\subseteq Z$, $|S|=2^s$, $0\le s\le \ell$ there exists a request sequence
$\sigma_S$ consisting of terminals from $X$ s.t.\
\begin{itemize}
\item for a node $x\in \sigma_S$ both its neighbors in $G$ are from set $S$;
\item there exists a node $z^*\in S$ with $\ex[\mydegb_{A,\sigma}(z^*)]\ge s/2$;
\item there exists an offline algorithm $\mathrm{OFF}$ for servicing requests
  in $\sigma$ with 
  $\max_{z\in
    S}\{\mydeg_{\mathrm{OFF},\sigma}(z)\}\le 1$, and
  $\mydeg_{\mathrm{OFF},\sigma}(z)=0$ for $z\in (Z\setminus S)\cup \{z^*\}$.
\end{itemize}
\end{lemma}
\begin{proof}
We prove the lemma by induction over $s$. The base case $s=0$ holds trivially
when choosing the empty request sequence. For the induction step consider an
arbitrary subset $S\subseteq Z$ with $|S|=2^{s+1}$. Partition $S$ into two
disjoint subsets $S_1$ and $S_2$ of cardinality $2^{s}$ each.

Let $\sigma_1$ be the request sequence that exists due to induction hypothesis
for set $S_1$. Hence, there is a node $z_1^*\in S_1$ with
$\ex[\mydeg'_{A,\sigma_1}(z_1^*)]\ge s/2$. Now, let $\tilde{A}$ behave like algorithm $A$
\emph{after} it already received a request sequence $\sigma_1$ (hence, it
starts with all edges that are chosen when running $A$ on $\sigma_1$;
note, however, that $\mydegb_{\tilde{A},\sigma_2}(\cdot)$ only takes into
account edges incident to nodes from $\sigma_2$). Due to
induction hypothesis for $\tilde{A}$ and set $S_2$ there exists a request sequence
$\sigma_2$ such that $\ex[\mydegb_{\tilde{A},\sigma_2}(z_2^*)]\ge s/2$ for a node $z_2^*\in
S_2$. Hence, the request sequence $\sigma=\sigma_1\circ\sigma_2$ fulfills
$\ex[\mydegb_{A,\sigma}(z_1^*)]\ge s/2$ and $\ex[\mydegb_{A,\sigma}(z_2^*)]\ge s/2$.

We extend the request sequence by appending the node $x$ that is connected to
$z_1^*$ and $z_2^*$ in $G$. After serving the request one of the edges
$\{x,z_1\}$ or $\{x,z_2\}$ must be chosen with probability at least 1/2 by $A$.
Hence, afterwards either $\ex[\mydegb_{A,\sigma}(z_1^*)]$ or
$\ex[\mydegb{\delta}_{A,\sigma}(z_2^*)]$ must be at least $(s+1)/2$.

It remains to argue that there exists a good offline algorithm. Combining the
offline algorithms $\mathrm{OFF}_1$ and $\mathrm{OFF}_2$ for $\sigma_1$ and
$\sigma_2$ gives an offline algorithm for $\sigma_1\circ\sigma_2$ that 
has $\max_{z}\{
\mydeg_{\mathrm{OFF},\sigma_1\circ\sigma_2}(z)\}\le 1$ 
and $\mydeg_{\mathrm{OFF},\sigma_1\circ\sigma_2}(z)\}=0$ for $z\notin S_1\cup
S_2$ and for $z\in\{z_1^*,z_2^*\}$. Now, when the node $x$ connected to $z_1^*$
and $z_2^*$ is appended to the request sequence the offline algorithm can
serve this request by either buying edge $\{x,z_1\}$ or $\{x,z_2\}$, and can
therefore gurantee that $z_*$ (either $z_1$ or $z_2$) has degree $0$.
\end{proof}

\medskip
\noindent
\begin{proof}[of Theorem~\ref{thm:LBst}]
	Choosing $s=\log n$ in the above lemma gives our lower bound.
\end{proof}


\newpage
\bibliographystyle{abbrv}
\bibliography{dbbib}

\begin{thebibliography}{10}

\bibitem{AKR91b}
A.~Agrawal, P.~Klein, and R.~Ravi.
\newblock When trees collide: An approximation algorithm for the generalized
  steiner problem on networks.
\newblock {\em SIAM Journal on Computing}, 24(3):440--456, 1995.

\bibitem{AKR91a}
A.~Agrawal, P.~N. Klein, and R.~Ravi.
\newblock {\em How Tough is the Minimum-degree Steiner Tree?: A New Approximate
  Min-max Equality}.
\newblock 1991.

\bibitem{AHL12}
S.~Alaei, M.~Hajiaghayi, and V.~Liaghat.
\newblock Online prophet-inequality matching with applications to ad
  allocation.
\newblock In {\em Proceedings of the 13th ACM Conference on Electronic
  Commerce}, pages 18--35, 2012.

\bibitem{AHL13}
S.~Alaei, M.~Hajiaghayi, and V.~Liaghat.
\newblock The online stochastic generalized assignment problem.
\newblock In {\em Approximation, Randomization, and Combinatorial Optimization.
  Algorithms and Techniques}, pages 11--25. 2013.

\bibitem{AHL11}
S.~Alaei, M.~T. Hajiaghayi, V.~Liaghat, D.~Pei, and B.~Saha.
\newblock Adcell: Ad allocation in cellular networks.
\newblock In {\em Algorithms--ESA 2011}, pages 311--322. 2011.

\bibitem{AAABN09}
N.~Alon, B.~Awerbuch, Y.~Azar, N.~Buchbinder, and J.~Naor.
\newblock The online set cover problem.
\newblock {\em SIAM Journal on Computing}, 39(2):361--370, 2009.

\bibitem{AAB96}
B.~Awerbuch, Y.~Azar, and Y.~Bartal.
\newblock On-line generalized steiner problem.
\newblock In {\em Proceedings of the seventh annual ACM-SIAM symposium on
  Discrete algorithms}, pages 68--74, 1996.

\bibitem{ABFP13}
Y.~Azar, U.~Bhaskar, L.~Fleischer, and D.~Panigrahi.
\newblock Online mixed packing and covering.
\newblock In {\em Proceedings of the Twenty-Fourth Annual ACM-SIAM Symposium on
  Discrete Algorithms}, pages 85--100. SIAM, 2013.

\bibitem{BKN09}
N.~Bansal, R.~Khandekar, and V.~Nagarajan.
\newblock Additive guarantees for degree-bounded directed network design.
\newblock {\em SIAM J. Comput.}, 39(4):1413--1431, 2009.

\bibitem{BC97}
P.~Berman and C.~Coulston.
\newblock On-line algorithms for steiner tree problems.
\newblock In {\em Proceedings of the twenty-ninth annual ACM symposium on
  Theory of computing}, pages 344--353, 1997.

\bibitem{BN09}
N.~Buchbinder and J.~Naor.
\newblock The design of competitive online algorithms via a primal: dual
  approach.
\newblock {\em Foundations and Trends{\textregistered} in Theoretical Computer
  Science}, 3(2--3):93--263, 2009.

\bibitem{CG82}
P.~M. Camerini and G.~Galbiati.
\newblock The bounded path tree problem.
\newblock {\em SIAM Journal on Algebraic Discrete Methods}, 3(4):474--484,
  1982.

\bibitem{CGM80}
P.~M. Camerini, G.~Galbiati, and F.~Maffioli.
\newblock Complexity of spanning tree problems: Part i.
\newblock {\em European Journal of Operational Research}, 5(5):346--352, 1980.

\bibitem{C73}
V.~Chv{\'a}tal.
\newblock Tough graphs and hamiltonian circuits.
\newblock {\em Discrete Mathematics}, 5(3):215--228, 1973.

\bibitem{EV14}
A.~Ene and A.~Vakilian.
\newblock Improved approximation algorithms for degree-bounded network design
  problems with node connectivity requirements.
\newblock {\em STOC}, 2014.

\bibitem{FMMM09}
J.~Feldman, A.~Mehta, V.~Mirrokni, and S.~Muthukrishnan.
\newblock Online stochastic matching: Beating 1-1/e.
\newblock In {\em Foundations of Computer Science, 2009. FOCS'09. 50th Annual
  IEEE Symposium on}, pages 117--126, 2009.

\bibitem{FR90}
M.~F{\"u}rer and B.~Raghavachari.
\newblock An {NC} approximation algorithm for the minimum degree spanning tree
  problem.
\newblock In {\em Allerton Conf. on Communication, Control and Computing},
  pages 274--281, 1990.

\bibitem{FR94}
M.~F{\"u}rer and B.~Raghavachari.
\newblock Approximating the minimum-degree steiner tree to within one of
  optimal.
\newblock {\em Journal of Algorithms}, 17(3):409--423, 1994.

\bibitem{garcia1988implementation}
H.~Garcia-Molina and B.~Kogan.
\newblock An implementation of reliable broadcast using an unreliable multicast
  facility.
\newblock In {\em Reliable Distributed Systems, 1988. Proceedings., Seventh
  Symposium on}, pages 101--111, 1988.

\bibitem{G06}
M.~X. Goemans.
\newblock Minimum bounded degree spanning trees.
\newblock In {\em Foundations of Computer Science, 2006. FOCS'06. 47th Annual
  IEEE Symposium on}, pages 273--282, 2006.

\bibitem{HLP14}
M.~Hajiaghayi, V.~Liaghat, and D.~Panigrahi.
\newblock Near-optimal online algorithms for prize-collecting steiner problems.
\newblock In {\em Automata, Languages, and Programming}, pages 576--587. 2014.

\bibitem{HLP13}
M.~T. Hajiaghayi, V.~Liaghat, and D.~Panigrahi.
\newblock Online node-weighted steiner forest and extensions via disk
  paintings.
\newblock In {\em Foundations of Computer Science (FOCS), 2013 IEEE 54th Annual
  Symposium on}, pages 558--567, 2013.

\bibitem{huang2003publish}
Y.~Huang and H.~Garcia-Molina.
\newblock Publish/subscribe tree construction in wireless ad-hoc networks.
\newblock In {\em Mobile Data Management}, pages 122--140, 2003.

\bibitem{IW91}
M.~Imase and B.~M. Waxman.
\newblock Dynamic {Steiner} tree problem.
\newblock {\em SIAM Journal on Discrete Mathematics}, 4(3):369--384, 1991.

\bibitem{KMT11}
C.~Karande, A.~Mehta, and P.~Tripathi.
\newblock Online bipartite matching with unknown distributions.
\newblock In {\em Proceedings of the forty-third annual ACM symposium on Theory
  of computing}, pages 587--596, 2011.

\bibitem{KVV90}
R.~M. Karp, U.~V. Vazirani, and V.~V. Vazirani.
\newblock An optimal algorithm for on-line bipartite matching.
\newblock In {\em Proceedings of the twenty-second annual ACM symposium on
  Theory of computing}, pages 352--358, 1990.

\bibitem{KKN13}
R.~Khandekar, G.~Kortsarz, and Z.~Nutov.
\newblock On some network design problems with degree constraints.
\newblock {\em Journal of Computer and System Sciences}, 79(5):725--736, 2013.

\bibitem{KR00}
J.~K{\"o}nemann and R.~Ravi.
\newblock A matter of degree: Improved approximation algorithms for
  degree-bounded minimum spanning trees.
\newblock In {\em Proceedings of the thirty-second annual ACM symposium on
  Theory of computing}, pages 537--546, 2000.

\bibitem{KR05}
J.~K{\"o}nemann and R.~Ravi.
\newblock Primal-dual meets local search: approximating msts with nonuniform
  degree bounds.
\newblock {\em SIAM Journal on Computing}, 34(3):763--773, 2005.

\bibitem{lau2011iterative}
L.~C. Lau, R.~Ravi, and M.~Singh.
\newblock {\em Iterative methods in combinatorial optimization}, volume~46.
\newblock Cambridge University Press, 2011.

\bibitem{LS13}
L.~C. Lau and M.~Singh.
\newblock Additive approximation for bounded degree survivable network design.
\newblock {\em SIAM Journal on Computing}, 42(6):2217--2242, 2013.

\bibitem{MY11}
M.~Mahdian and Q.~Yan.
\newblock Online bipartite matching with random arrivals: an approach based on
  strongly factor-revealing {LP}s.
\newblock In {\em Proceedings of the forty-third annual ACM symposium on Theory
  of computing}, pages 597--606, 2011.

\bibitem{MOS12}
V.~H. Manshadi, S.~Oveis~Gharan, and A.~Saberi.
\newblock Online stochastic matching: Online actions based on offline
  statistics.
\newblock {\em Mathematics of Operations Research}, 37(4):559--573, 2012.

\bibitem{MRSR98}
M.~V. Marathe, R.~Ravi, R.~Sundaram, S.~Ravi, D.~J. Rosenkrantz, and H.~B.
  Hunt~III.
\newblock Bicriteria network design problems.
\newblock {\em Journal of Algorithms}, 28(1):142--171, 1998.

\bibitem{M12}
A.~Mehta.
\newblock Online matching and ad allocation.
\newblock {\em Theoretical Computer Science}, 8(4):265--368, 2012.

\bibitem{MSVV07}
A.~Mehta, A.~Saberi, U.~Vazirani, and V.~Vazirani.
\newblock Adwords and generalized online matching.
\newblock {\em Journal of the ACM (JACM)}, 54(5):22, 2007.

\bibitem{GJ79}
R.~G. Michael and S.~J. David.
\newblock Computers and intractability: a guide to the theory of
  np-completeness.
\newblock {\em WH Freeman \& Co., San Francisco}, 1979.

\bibitem{MOZ12}
V.~S. Mirrokni, S.~Oveis~Gharan, and M.~Zadimoghaddam.
\newblock Simultaneous approximations for adversarial and stochastic online
  budgeted allocation.
\newblock In {\em Proceedings of the Twenty-Third Annual ACM-SIAM Symposium on
  Discrete Algorithms}, pages 1690--1701, 2012.

\bibitem{NPS11}
J.~Naor, D.~Panigrahi, and M.~Singh.
\newblock Online node-weighted steiner tree and related problems.
\newblock In {\em Foundations of Computer Science (FOCS), 2011 IEEE 52nd Annual
  Symposium on}, pages 210--219, 2011.

\bibitem{N12}
Z.~Nutov.
\newblock Degree-constrained node-connectivity.
\newblock In {\em LATIN 2012: Theoretical Informatics}, pages 582--593. 2012.

\bibitem{PY82}
C.~H. Papadimitriou and M.~Yannakakis.
\newblock The complexity of restricted spanning tree problems.
\newblock {\em Journal of the ACM}, 29(2):285--309, 1982.

\bibitem{QW11}
J.~Qian and D.~P. Williamson.
\newblock An o (logn)-competitive algorithm for online constrained forest
  problems.
\newblock In {\em Automata, Languages and Programming}, pages 37--48. 2011.

\bibitem{RT85}
P.~Raghavan and C.~D. Thompson.
\newblock Provably good routing in graphs: regular arrays.
\newblock In {\em Proceedings of the seventeenth annual ACM symposium on Theory
  of computing}, pages 79--87, 1985.

\bibitem{LS07}
M.~Singh and L.~C. Lau.
\newblock Approximating minimum bounded degree spanning trees to within one of
  optimal.
\newblock In {\em Proceedings of the thirty-ninth annual ACM symposium on
  Theory of computing}, pages 661--670, 2007.

\end{thebibliography}

\newpage

\appendix
\section{Figures}

\begin{algorithm}[!h]
	\textbf{Input:} A graph $G$, and an online stream of demands $(s_1,t_1), (s_2,t_2), \ldots$.
	
	\textbf{Output:} A set $H$ of edges such that every given demand $(s_i,t_i)$ is connected via $H$.
	
	\textbf{Offline Process:}
		\begin{algorithmic}[1]
			\STATE Initialize $H=\varnothing$.
		\end{algorithmic}
	
	\textbf{Online Scheme; assuming a demand $(s_i,t_i)$ is arrived:}
	
	\begin{algorithmic}[1]
		\STATE Compute $P_i$, a $(s_i, t_i)$-path with the smallest uptick load $\ell^+_H(P_i^*)$ in its extension part.
		\begin{itemize}
			\item Shortcut the connected components of $H$ by replacing the edges of a component with that of a clique.
			\item In the resulting graph, define the distance of a vertex $v$ from $s_i$ as the minimum uptick load of a $(s_i,v)$-path.
			\item Find $P_i$ by evoking a Dijkstra-like algorithm according to this notion of distance. 
		\end{itemize}
		\STATE Set $H=H\cup P_i^*$.
	\end{algorithmic}
	\caption{Online Degree-Bounded Steiner Forest}
	\label{alg:dbsf}
\end{algorithm}

\begin{figure}[!h]
	\begin{center}
		\includegraphics[width=0.7\textwidth]{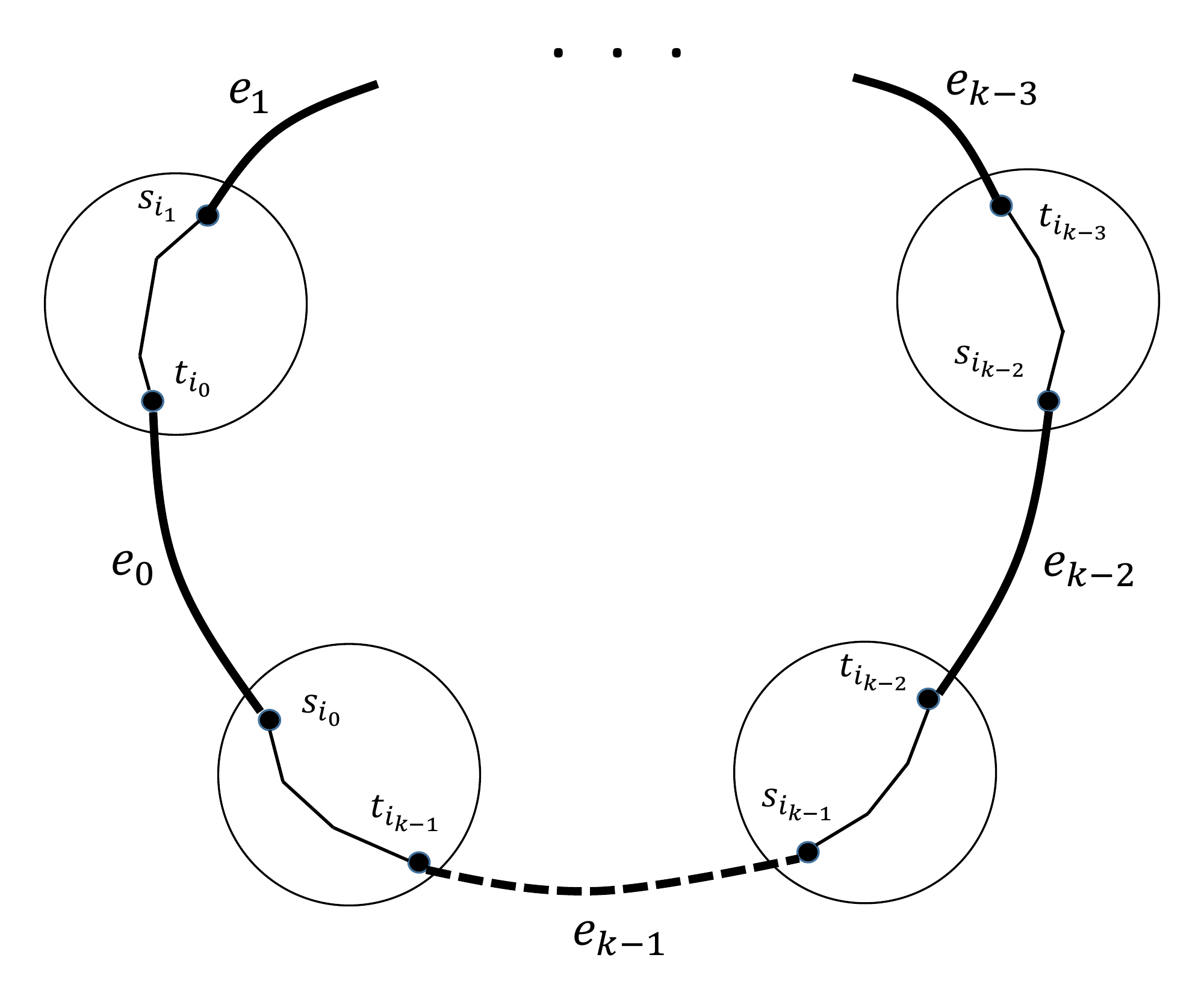}
	\end{center}
	\caption{Existence of a cycle implies that of a path with low uptick load in its extension part.}
	\label{fig:lemsep}
\end{figure}

\section{Lower Bounds for Other Degree-Bounded Steiner Connectivity Problems}
In this section we consider strong lower bounds for the general form of two variants of degree-bounded network design problems.
\subsection{{\sc Online Degree-Bounded Edge-weighted Steiner Tree}}
Below we present a graph instance $G=(V,E)$ for online bounded-degree
edge-weighted Steiner tree in which for any (randomized) online algorithm $A$
there exists a demand sequence for which $A$ either violates the degree bound
by a large factor or generates a much larger weight than the optimum.


Consider a graph $G$ as shown in Figure~\ref{fig:1}, which has $n=2k+1$ nodes
and $3k$ edges. Every node $i$ ($1\leq i\leq k$) is connected to the root with
a zero-weight edge and to node $k+i$ with weight $n^i$. In addition, there
exist zero-weight edges connecting node $i$ ($k+1 \leq i < 2k$) to node $i+1$,
and there exists a zero-weight edge that connects node $2k$ to the root.
 We assume that all node weights $b_v$ are equal to one and the degree bound $b$
 for $\OPT$ is equal to $3$.
\begin{figure*}[!h]
	\centering
	
	\begin{subfigure}[b]{0.44\textwidth}
		\includegraphics[width=\textwidth]{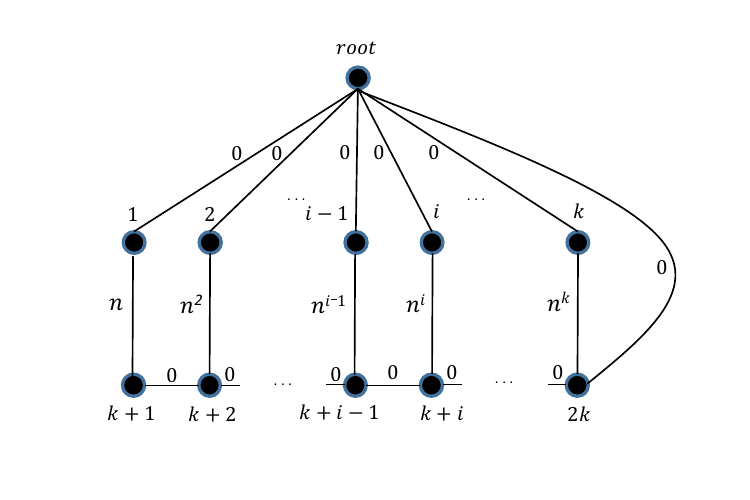}
		\caption{The graph $G$ consists of $2k+1$ nodes and $3k$ edges.}
		\label{fig:1}
	\end{subfigure}%
	\hspace{0.1\textwidth}
	\begin{subfigure}[b]{0.44\textwidth}
		\includegraphics[width=\textwidth]{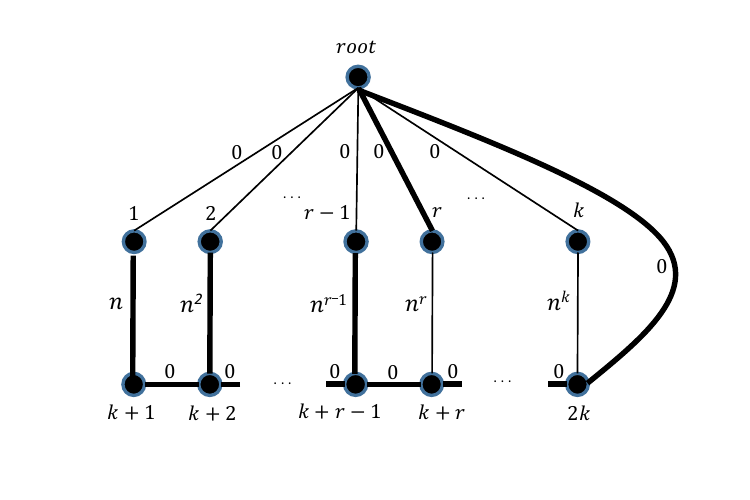}
		\caption{The highlighted subtree represents an optimum solution $\OPT_3$.}
		\label{fig:2}
	\end{subfigure}%
	\caption{The hard example for edge-weighted Steiner tree}
\end{figure*}

\begin{proof}[of Theorem~\ref{thm:ewLB}]
  The adversary consecutively presents terminals starting from node $1$. At
  each step $i$ the algorithm $A$ adds some edges to its current solution
  such that the $i^{th}$ terminal $t_i=i$ gets connected to the $\mathit{root}$.

  We use $X_i$ to denote the subset of edges chosen by $A$ after step $i$. We also use
  $0\leq p_{ij}\leq 1$ and $0\leq q_{ij}\leq 1$ to denote the probability that
  the edges $\{i,\mathit{root}\}$ and $\{i,k+i\}$, respectively, are in $X_j$ . After each
  step $j$, all terminals $t_i$ ($i\leq j$) must be connected to the $\mathit{root}$ by
  at least one of $\{i,\mathit{root}\}$ or $\{i,k+i\}$. Therefore $p_{ij}+q_{ij}\geq 1$.
  In addition, for every $j_1\leq j_2$ we have $q_{ij_1}\leq q_{ij_2}$ and
  $p_{ij_1}\leq p_{ij_2}$, because $X_{j_1} \subseteq X_{j_2}$.
  The adversary stops the sequence at the first step $r$ for which $q_{rr}>{1}/{2}$.
  If this never happens the sequence is stopped after requesting $k$ nodes.


  We use $\OPT_b$ to denote the weight of a minimum Steiner tree with maximum
  degree $b=3$. In order to find an upper bound for $\OPT_b$, consider the
  following tree $T$:
\begin{align*}
T&=\{\{i,k+i\}|1\leq i\leq r-1\} \cup \{\{i,i+1\}|k\leq i<2k\} \cup \{\{2k,\mathit{root}\},\{r,\mathit{root}\}\}\enspace.
\end{align*}
As we can observe in Figure~\ref{fig:2}, $T$ meets the degree bounds and
connects the first $r$ terminals to the root, so forms a valid solution. We have
\begin{equation*}
w(\OPT) \leq w(T)=\sum_{i=1}^{r-1} n^i \in O(n^{r-1})\enspace.
\end{equation*}

\noindent
Back to algorithm $A$, the adversary causes one of the following two
cases: 
\begin{enumerate}[leftmargin=*,label=(Case~\alph*)]
\item
The process stops at step $r$ with
$q_{rr}>\frac{1}{2}$. Then
\begin{align*}
\ex[w(X_r)]&=\sum_{i=1}^{r} q_{ir}\cdot n^i \geq q_{rr}\cdot n^r\geq \frac{n^r}{2}\enspace.
\end{align*}
Hence, $\ex[w(X)]\ge\Omega(n)\cdot w(\OPT)$.
\item
The process continues until step $k$, i.e., for every $i\leq k$, $q_{ii}\leq\frac{1}{2}$.
\begin{align*}
\ex[\operatorname{deg}(\mathit{root})] 
&\ge\sum_{i=1}^{k}p_{ik}\ge\sum_{i=1}^{k}p_{ii}\geq\sum_{i=1}^{k}(1-q_{ii})\ge\frac{k}{2}=\frac{n-1}{4}\enspace.
\end{align*}
Hence, $\ex[\operatorname{deg}(\mathit{root})]\ge  \Omega(n)\cdot b$.
\end{enumerate}
This shows that for any online algorithm $A$ there is a demand sequence on
which $A$ either generates a large weight or violates the degree bound by a
large factor.
%
%
\end{proof}

\subsection{{\sc Online Degree-Bounded Group Steiner Tree}}
In this section, presenting an adversary scenario, we show there is no deterministic algorithm for \odbgst{} with competitive ratio $o(n)$ even if $G$ is a star graph.

\begin{proof}[of Theorem~\ref{thm:gstLB}]
	For any integer $n>1$, we provide a graph instance $G$ of size $n$ and an online scenario, in which no deterministic algorithm can obtain a competitive ratio better than $n-1$.
	Let $G$ be a star with $n-1$ leaves $v_2$ to $v_{n}$, and $v_{1}$ be the internal node. For an algorithm $\mathcal{A}$, we describe the adversary scenario as follows.
	
	Let $v_{1}$ be the root. Let the first demand group be the set of all leaves. Whenever $\mathcal{A}$ connects a node $v_i$ to $v_1$ in $H$, adversary removes the selected nodes $v_i$ from the next demand groups, until all leaves are connected to $v_1$ in $H$. In particular let $C$ denote the set of all nodes connected to $v_1$ in $H$ so far. Let the demand group be the set of all leaves in $\{v_2, v_3, \ldots, v_n\} \setminus C$. We do this until $\{v_2, v_3, \ldots, v_n\} \setminus C = \emptyset$, which means every leaf is connected to $v_1$ in $H$. $G$ is a star, thus a leaf $v_i$ is connected to $v_1$ in $H$ \text{iff} $H$ includes the edge between $v_i$ and $v_1$. Thus after all demands $\mathcal{A}$ has added all edges in $G$ to $H$. Hence $deg(v_1)=n-1$.
	
	Now consider the optimal offline algorithm. Let $g_i$ denote the $i$-th demand group. Assume we have $k$ demand groups. By construction of the demand groups $g_{k} \subset g_{k-1} \subset \ldots \subset g_{1}$. Thus there is a single node that exists in all group demands. Hence the optimal offline algorithms only needs to connect that node to the root in $H$. Therefore, the degree of each node is at most $1$ and the competitive ratio is $n-1$.
\end{proof}

\end{document}